\documentclass[11pt]{article}
\usepackage[utf8]{inputenc}
\usepackage{amsmath,amsthm,amssymb,caption,subcaption,nicefrac}
\usepackage{fullpage}
\usepackage{bm}

\usepackage{mathtools}
\usepackage{xspace}
\usepackage{xcolor}
\usepackage{colortbl}
\usepackage{soul}
\usepackage{hyperref}
\hypersetup{
    colorlinks=true,
    linkcolor=olive,
    citecolor=purple
}
\usepackage{t1enc}
\usepackage[nameinlink]{cleveref}
\newcommand{\etal}{\emph{et~al.}}

\definecolor{cellcolor}{rgb}{1,0.8,0.8}
\newcommand{\cellcheck}{\cellcolor{cellcolor}\checkmark}

\DeclareMathOperator{\poly}{poly}
\DeclareMathOperator*{\argmax}{arg\,max}
\DeclareMathOperator{\score}{score}

\newcommand{\defeq}{\coloneq}
\newcommand{\opt}{\text{OPT}}
\newcommand{\temp}{\text{temp}}

\newcommand{\matrixapp}[2]{
  \begin{bmatrix}
    #1  \\
    #2      
  \end{bmatrix}
}

\newtheorem{theorem}{Theorem}
\newtheorem{lemma}[theorem]{Lemma}
\newtheorem{proposition}[theorem]{Proposition}
\newtheorem{corollary}[theorem]{Corollary}

\newtheorem{definition}[theorem]{Definition}

\newcommand{\probstyle}[1]{\textsc{#1}\xspace}
\newcommand{\wthiele}{\probstyle{$w$-Thiele}}
\newcommand{\thielearg}[1]{\probstyle{#1-Thiele}}
\newcommand{\genthiele}{\probstyle{Generalized Thiele}}
\newcommand{\pavfull}{\probstyle{Proportional Approval Voting}}
\newcommand{\pav}{\probstyle{PAV}}
\newcommand{\ccavfull}{\probstyle{Chamberlin-Courant Approval Voting}}
\newcommand{\cc}{\probstyle{CC}}
\newcommand{\avfull}{\probstyle{Multi-winner Approval Voting}}
\newcommand{\av}{\probstyle{AV}}
\newcommand{\ellcoverage}{\probstyle{$\ell$-Coverage}}

\title{Algorithms for Structured Elections under Thiele Voting Rules\footnote{A conference version of this work appeared in AAAI 2026~\cite{LassotaS2026_pavvi}.}}\author{
    Alexandra Lassota\\
    TU Eindhoven, the Netherlands\\
    {\small\texttt{a.a.lassota@tue.nl}}\\[-2pt]
    \and
    Krzysztof Sornat\\
    AGH University, Poland\\[-2pt]
    {\small\texttt{sornat@agh.edu.pl}}
}
\date{}

\begin{document}

\maketitle

\begin{abstract}
    We study the computational complexity of winner determination problems in approval-based committee elections under Thiele voting rules.
    These form a class of rules parameterized by a fixed weight vector that specifies how a voter's satisfaction depends on the number of approved candidates elected.
    We first analyze the structure of optimal solutions based on the sets of voters who approve each candidate---that is, how voters' approval ballots induce dependencies between candidates---revealing constraints on a winning committee under any fixed Thiele voting rule.
    Using this, we design FPT algorithms for Proportional Approval Voting (PAV) and other Thiele rules on a natural restricted domain known as the Voter Interval (VI) domain---that is, after a suitable ordering of voters, each candidate is approved by a consecutive interval of voters.
    In particular, we show that every Thiele rule on VI is FPT with respect to a parameter for which the problem is NP-hard on general instances, even when the parameter takes constant values.
    Our results advance the understanding of the computational complexity of PAV on Voter Interval instances, which remains one of the central open questions in this area.
    We further resolve two open questions from the literature on PAV (and other Thiele voting rules)
    by providing a polynomial-time algorithm for instances where each candidate is approved by at most two voters, and an FPT algorithm parameterized by the total score of a winning committee.
\end{abstract}

\section{Introduction}

Multi-winner elections based on approval ballots are used in many settings, such as recommendation systems, committee selection, and blockchain~\cite{SkowronFL16_aij_set_of_items,lackner23abc_book,BoehmerBCG0S24_polkadot}, where
every voter expresses its preferences as a subset of candidates it approves of.
A central family of voting rules for such applications are Thiele rules~\cite{Thiele95},
which include Proportional Approval Voting (PAV) as a prominent member~\cite{AzizGGMMW15}.
These rules balance diversity, proportionality, and excellence in the selected committee, depending on the specific Thiele rule used~\cite{LacknerS21}.
For instance, Chamberlin-Courant Approval Voting (CC)~\cite{ChamberlinCourant83} promotes diversity by ensuring broad representation; PAV aims for proportional representation, satisfying strong proportionality axioms like EJR+~\cite{BrillP23_ejrp};
Multi-winner Approval Voting focuses on excellence by selecting the most approved candidates.
However, computing a winning committee under Thiele rules is computationally challenging.
In fact, finding a winning committee under any non-trivial Thiele rule
(i.e., every rule except for Multi-winner Approval Voting)
is NP-hard.
For CC and PAV, the problems remain NP-hard even under very restricted conditions where each candidate is approved by exactly 3 voters and each voter approves exactly 2 candidates~\cite{AzizGGMMW15,SkowronFL16_aij_set_of_items}.

To better understand and overcome this computational hardness,
a natural direction is to restrict the input domain.
Such restrictions often enable polynomial-time algorithms for otherwise intractable rules,
and this approach has been particularly fruitful in approval-based committee elections.
Two central restricted domains for approval ballots are the \emph{Candidate Interval (CI)} and the \emph{Voter Interval (VI)} domains~\cite{ElkindLP2017trends_book,ElkindLP25_arxiv_preference_restrictions}.

In the CI domain, candidates can be ordered so that each voter's approval set forms a contiguous interval.
CI captures scenarios where candidates are linearly ordered, e.g. by ideology or location,
and each voter is focused on a specific region of the spectrum.
One key advantage of CI is that it allows the winner determination problem for all Thiele rules to be solved efficiently.
It is by using an integer linear programming formulation that admits a totally unimodular constraint matrix,
which are known to be solvable in polynomial time~\cite{Peters18_aaai,PetersL20_spoc_jair}.

The VI domain, in contrast, imposes structure on the voters rather than on the candidates.
Here, voters can be ordered so that each candidate is approved by a consecutive segment of voters.
This domain models scenarios where voters are structured by demographic or socioeconomic features such as age, income, or education level,
and candidates appeal to specific groups, e.g. young voters, low-income households, or university-educated individuals.
The VI domain has received significant attention~\cite{ElkindLP2017trends_book,ElkindLP25_arxiv_preference_restrictions}
and, consequently, for several voting rules, such as CC, Monroe's and Minimax Approval Voting
polynomial-time winner-determination algorithms for VI elections are known~\cite{BetzlerSU13_jair,LiuG16_mav_vi}.

Yet, for most Thiele rules, and PAV in particular, the computational complexity of winner determination on VI remains a prominent open question~\cite{ElkindLP2017trends_book,Peters18_aaai,GodziszewskiB0F21_vci_aaai21,lackner23abc_book,ElkindLP25_arxiv_preference_restrictions}.
Despite the dual nature of the VI and CI domains, the techniques used for CI do not translate to VI.
In particular, the constraint matrix of the integer linear programming formulation is not totally unimodular anymore under VI preferences and, thus,
it is not clear whether this problem can be solved in polynomial time or not. 
This motivates a search for new structural results and algorithmic techniques.

\paragraph{Structural Results.}
This work takes a new approach to understand the computational complexity of Thiele rules
by studying the structure of winning committees.
Instead of considering candidates in isolation,
we examine how shared support among voters constrains the possible combinations of candidates in a winning committee.
We provide a structural characterization of optimal committees under Thiele rules based on a \emph{dominance} relation among candidates.
Specifically, a candidate $c$ is said to dominate another candidate $d$ if the set of supporters of $c$ strictly contains that of $d$.
This induces a hierarchy of \emph{dominancy levels} where candidates within the same level do not dominate one another and each is dominated by some candidate in a higher level.
We show that there always exists a winning committee that is \emph{non-dominated},
meaning that no member is dominated by any candidate outside the committee.
These structural insights apply to general approval profiles and are of independent interest, offering new theoretical tools for analyzing and determining winners under Thiele rules.
Practically, this structure can guide the development of more efficient algorithms and heuristics.
Theoretically, it offers a new perspective on the open question of whether winner determination under PAV is polynomial-time solvable on VI.

\paragraph{Algorithmic Results.}
We show that our structural insights are particularly effective when combined with the VI property:
candidates of a VI instance can be partitioned into parts such that candidates within a part influence only a limited number of voters.
This allows to design a dynamic program over a sequence of such parts, exploiting the limited interaction.
As a result, we obtain an FPT algorithm for PAV and, more generally, for any Thiele rule, when parameterized by two parameters combined:
the maximum number of approvals received by a candidate, and
the maximum number of approvals in a vote.
Notably, the same parameterization is para-NP-hard in the general (unstructured) case---that is, the problem remains NP-hard even when both parameters are constants.
This contrast highlights the algorithmic power of our structural approach and provides substantial progress toward resolving the open question on the complexity of PAV in the VI domain~\cite{ElkindLP2017trends_book,Peters18_aaai,lackner23abc_book,ElkindLP25_arxiv_preference_restrictions}.

Beyond the VI setting, we further contribute to understanding tractable cases for Thiele rules by resolving two open problems posed by Yang and Wang~\cite{YangW18-aamas18,YangW23-jaamas}, originally asked for PAV:
(1) we provide a polynomial-time algorithm for instances where each candidate is approved by at most two voters, based on a proper integer linear programming formulation; and
(2) we give an FPT algorithm parameterized by the total score of a winning committee, employing combinatorial tools like color-coding and splitters.
It is particularly interesting as the score may not be integer.
Moreover, this parameter can be smaller than the number of voters for which FPT algorithms exist.
Importantly, our results extend to all Thiele rules, not just PAV, demonstrating the generality of our approach.

\paragraph{Structure of the Paper.}
\Cref{sec:preliminaries} introduces notation for the approval-based committee election model, structured domains and voting rules studied in this paper.
\Cref{sec:structure} presents our structural results on winning committees under Thiele rules.
\Cref{sec:fpt-vi}, using the structural results, develops an FPT algorithm on VI instances.
In \Cref{sec:general-instances}, we give an FPT algorithm parameterized by the total score of the optimal committee, and a polynomial-time algorithm for instances where each candidate is approved by at most two voters.

\subsection{Related Work}\label{subsec:relatedwork}
Recently, another FPT algorithm parameterized by the total score of a winning committee was provided independently by Gupta, Jain, Saha, Saurabh and Upasana~\cite{GuptaJSSU2025_independent_fpt}.
Even though they also use color-coding at the heart of their algorithm, the approaches differ significantly: while they color candidates and voters, we color the approvals of the voters.
This results in a different number, type, and meaning of guesses, as well as a different construction of the overall solution.
It also results in differences in running times.
While the algorithms outperform each other in some cases, our algorithm runs in time truly linear in the number of voters, which we see as a strong advantage.
For more detailed discussion about differences between both algorithms, as well as quantitative comparison of running times, we defer to \Cref{sec:fpt-d}.

There is also a substantial body of work on tractability of voting rules under restricted domains in the ordinal setting (in which voters cast votes in a form of linear orders over candidates),
such as Single-Peaked (SP)
and Single-Crossing (SC) preferences.
These domains are the ordinal counterparts of CI and VI, respectively,
and have led to efficient algorithms for many voting rules, including CC (in its general version defined on cardinal values of misrepresentation) and Kemeny.
In particular, a classic dynamic programming algorithm was proposed for CC by Betzler \etal~\cite{BetzlerSU13_jair}.
More recently, an algorithm with near linear-time in the input size has been developed for the case where
the SC-axis (ordering of voters) 
or the SP-axis (ordering of candidates) is explicitly given~\cite{ConstantinescuE21_sc_linear,SornatWX22_ijcai}.

These works illustrate the general principle that structural properties of preferences can be algorithmically exploited.
For more information about these and other restricted domains see, e.g., the works of Elkind, Lackner and Peters~\cite{ElkindLP2017trends_book,ElkindLP25_arxiv_preference_restrictions}.

\paragraph{Parameterized Complexity of Thiele Rules.}
For an overview of the parameterized complexity of Thiele rules,
we refer to the recent work by Yang and Wang~\cite{YangW23-jaamas},
who provide a comprehensive summary of known results (see Table~1 therein),
along with new findings for several structural parameters and their combinations.
Below, we briefly discuss the main results for standard parameters such as the number of voters~$n$, the number of candidates~$m$, and the committee size~$k$, with an emphasis on the techniques used and known lower bounds. 

A trivial brute-force algorithm checking total scores of all size-$k$ subsets of candidates runs in FPT time with respect to $m$, namely $O^*(2^m)$.
This is essentially optimal for every non-constant Thiele rule, as under the Exponential Time Hypothesis there is no $O^*(2^{o(m)})$-time algorithm for this problem~\cite{SornatWX22_ijcai}.

A mixed integer linear program (MILP) presented by Faliszewski \etal~\cite[Fig.~2]{FaliszewskiSST18_scw}
implies an FPT algorithm with respect to $n$ for every Thiele rule.\footnote{
    Even though ordinal ballots and top-$k$-counting rules are considered in the MILP~\cite[Fig.~2]{FaliszewskiSST18_scw}, it is enough to replace $g_{m,k}(j)$ with $w_j$, where $w_j$ is the $j$-th element of the Thiele sequence $w$, and to adjust the definition of $\mathcal{T}(S_i)$ to the set of candidates approved by voters from $S_i$ and no voter from $V \setminus S_i$.
}
The main idea is to define an integer variable for each candidate type (defined by its set of supporters; hence, there are at most $2^n$ types), which encodes how many candidates of a particular type are selected for the solution.
Non-integral variables are forced to take integral values in the optimal solution, as first used by Bredereck \etal~\cite{BredereckFNST15_adt}.
The resulting running time is double-exponential, namely $O^*(2^{2^{O(n)}})$~\cite{BredereckF0KN20_aaai}.
A single-exponential lower bound of $O^*(2^{o(n)})$ under the Exponential Time Hypothesis is known~\cite{SornatWX22_ijcai},
hence, there is still a significant gap remaining.
A similar MILP idea has been applied in many contexts, e.g., for extensions of Thiele rules~\cite{JainST20_ijcai,YangW23-jaamas},
as well as for problems related to bribery and control in elections~\cite{BredereckFNST20_tcs}.

Every non-constant Thiele rule is also W[1]-hard with respect to the committee size~\cite{AzizGGMMW15,JainST20_ijcai,SornatWX22_ijcai}.

\paragraph{Approximability of Thiele Rules.}
Thiele rules have also been studied from the perspective of approximability,
where the goal is to compute a committee whose total score is close to optimal~\cite{SkowronFL16_aij_set_of_items,ByrkaSS18_hkm}.
For a broad class of Thiele rules, tight polynomial-time approximation algorithms have been established~\cite{DudyczMMS20_ijcai,BarmanFF21_concave_coverage,BarmanFGG22_lcoverage}.

\section{Preliminaries}\label{sec:preliminaries}

We are given a set of candidates $C = \{c_1, c_2, \dots, c_m\}$ and a set of voters $V = \{v_1, v_2, \dots, v_n\}$.
Each voter $v \in V$ expresses its preference in the form of an \emph{approval set} $A_v \subseteq C$, and the collection $A = (A_v)_{v \in V}$ is referred to as an \emph{approval profile}.
Any subset of $C$ is called a \emph{committee}.
We write $\mathcal{W}_k \defeq \{W \subseteq C: |W| = k \}$ to denote the set of all committees of size $k$.
For a candidate $c \in C$, we denote by $V_c \defeq \{ v \in V: c \in A_v\}$ the \emph{set of supporters} of $c$.
We extend the notation to sets of candidates $C' \subseteq C$, i.e., $V_{C'} \defeq \bigcup_{c \in C'} V_c$.
We define $\Delta_C \defeq \max_{c \in C} |V_c|$ as the maximum number of approvals given to a candidate
and $\Delta_V \defeq \max_{v \in V} |A_v|$ as the maximum number of approved candidates by a voter.
An \emph{approval-based committee (ABC) election} is a tuple $E=(C,V,A,k)$.
A voting rule is a function taking an election as an input and outputs a set of \emph{winning committees} of size $k$.

The $w$-Thiele voting rule~\cite{Thiele95} is parameterized by a non-increasing infinite sequence $w = (w_1, w_2, \dots)$, called a \emph{Thiele sequence}.
Given an election $(C,V,A,k)$,
a committee $W \in \mathcal{W}_k$ is \emph{optimal} under the $w$-Thiele rule if it maximizes the \emph{total score}:
\begin{align*}
    d \defeq \score_w(W) = \sum_{v \in V} \sum_{i = 1}^{|A_v \cap W|} w_i,
\end{align*}
over all committees of size $k$.
The $w$-Thiele rule returns the set of all such optimal committees.
The corresponding computational problem \wthiele requires outputting a single optimal committee.
Its decision variant asks whether there exists a committee of size $k$ with a total score at least a given value.

All Thiele rules considered in the literature are defined with $w_1 = 1$ (see examples below).
Dividing $w$ by $w_1$ and obtaining $w_1=1$ does not affect either its set of optimal solutions or its approximability~\cite{DudyczMMS20_ijcai},
but it affects the parameter $d$, the total score of an optimal solution that is studied in this paper.
In particular, if we allow $w_1 \leq \frac{1}{k}$, then $d \leq \Delta_C \cdot k \cdot \frac{1}{k} = \Delta_C$, but every \wthiele with non-constant~$w$ is NP-hard even if $\Delta_C = 3$~\cite{AzizGGMMW15,SkowronFL16_aij_set_of_items},
so this would imply paraNP-hardness with respect to $d$.
Therefore, we use the standard normalization $w_1 = 1$ in this paper.
(However any algorithm with running time dependent on $d$
after such normalization becomes an algorithm with running time dependent on $d + \frac{1}{w_1}$.)

Arguably, three most prominent $w$-Thiele rules are:
\begin{itemize}
    \item \ccavfull (\cc), that is \wthiele with $w = (1,0,0,\dots)$.
    \item \pavfull (\pav), that is \wthiele with $w_j = \frac{1}{j}$.
    \item \avfull (\av), that is \wthiele with $w = (1,1,1,\dots)$.
\end{itemize}
An interpolation between CC and AV is \ellcoverage
that is \wthiele with $w_j = 1$ for $j \leq \ell$ and $w_j = 0$ otherwise~\cite{BarmanFGG22_lcoverage}.

\genthiele is a voting rule \cite{SornatWX22_ijcai}
which takes as an input an election $(C,V,A,k)$ and $n$ Thiele sequences represented as $w \colon V \times \mathbb{N} \rightarrow [0,1]$, where $(w^v_i)_{i \in \mathbb{N}}$ is a Thiele sequence for a voter $v \in V$, and outputs a committee $W \in \mathcal{W}_k$ which maximizes the total score:
$
    d \defeq \score_w(W) = \sum_{v \in V} \sum_{i = 1}^{|A_i \cap W|} w^v_i,
$
over all committees of size $k$.
Naturally, \genthiele where $w^{v} = w^{v'}$ for $v,v' \in V$ is equivalent to \thielearg{$w^v$}.

\begin{definition}[Voter Interval]
    An approval profile $A$ has Voter Interval (VI) property if there exists a linear order of voters such that for every candidate $c \in C$, the set $V_c$ is an interval on the linear order.
\end{definition}

An analogous restricted domain is defined for ordering of candidates.
\begin{definition}[Candidate Interval]
    An approval profile $A$ has Candidate Interval (CI) property if there exists a linear order of candidates such that for every voter $v \in V$, the set $A_v$ is an interval on the linear order.
\end{definition}
If an approval profile is VI, then a corresponding ordering of voters---called a \emph{VI-axis}---can be found in polynomial time;
thus, we assume w.l.o.g. that the ordering is $(v_1, v_2, \dots v_n)$.
An analogous ordering for CI profiles can be found also in polynomial time~\cite{FaliszewskiHHR11,ElkindL15_ijcai15}.

We write $[n] = \{1,2,\dots,n\}$ and adopt standard notation from computational and parameterized complexity theory~\cite{CyganFKLMPPS15_fpt_book}.
In particular, a decision problem parameterized by $k$ is \emph{fixed parameter tractable} (FPT) with respect to $k$
if it can be solved in time $f(k) \cdot \poly(|\mathcal{I}|)$ for any instance $(\mathcal{I},k)$,
where $f$ is a computable function and $|\mathcal{I}|$ denotes the input size.
A problem solvable in time $|\mathcal{I}|^{f(k)}$ belongs to the class XP, which implies it is solvable in polynomial time for any fixed value of $k$.
The notation $O^*(\cdot)$ suppresses factors polynomial in the input size.

\section{Structure of Winning Committees}\label{sec:structure}

For a given approval profile $A$, we create the \emph{dominancy graph} of $A$ where vertices correspond to candidates and a directed edge from $c$ to $c'$ exists if $V_{c'} \subset V_c$, i.e., a candidate $c'$ is dominated by $c$ (all supporters of $c'$ are also supporters of $c$, and $c$ has a supporter not supporting $c'$).
This dominancy relation coincides with the one recently introduced independently by Dong \etal~\cite{DongBWBE25_aamas25}.

Non-dominated candidates are those candidates that have no incoming edge in the dominancy graph, we denote this set as $L_1$.
We partition $C$ into \emph{dominancy levels} $L_1, \dots, L_\delta$ such that every candidate $c \in L_i$ is at distance exactly $i$ in $G$ from some candidate in $L_1$ and there is no candidate $c' \in L_1$ such that the distance in the dominancy graph between $c'$ and $c$ is strictly smaller than $i$.
We denote by $\delta$ the \emph{depth} of a dominancy graph.
By the definition of dominancy relation involving strict inclusion, we have $\delta \leq \Delta_C+1$.

\begin{definition}
  A committee $W \subseteq C$ is \emph{non-dominated} if every candidate belonging to any directed path from an element of $L_1$ to an element of $W$ also belongs to $W$.
\end{definition}

The next theorem characterizes optimal solutions to \genthiele with respect to non-dominancy.

\begin{theorem}\label{thm:opt-non-dominated}
    There exists an optimal solution to \genthiele that is non-dominated.
    Furthermore,
    if $w_i^v > w_{i+1}^v$ for all $v \in V$ and all $i \in \mathbb{N}$,
    then every optimal solution to \genthiele must be non-dominated.
\end{theorem}
\begin{proof}
    Let $E$ be an instance of \genthiele.
    Let $W$ be an optimal solution to \genthiele on $E$.
    We can modify $W$ until it becomes non-dominated:
    considering a directed path $P = (p_1, \dots, p_{|P|})$ of maximum length such that $p_1 \in L_1, p_{|P|} \in W$ and $p_i \notin W$ for some $i \in \{1, \dots,|P|-1 \}$, we replace $W$ with $(W \setminus P) \cup \{p_i: i \in \{1, \dots, |W \cap P| \}$ which has the same cardinality as $W$. We repeat this procedure until $W$ is non-dominated. 
    
    This must happen after a finite number of steps, as by considering paths $P$ of maximum length $\ell \in [\delta]$, after at most $k$ steps (considering at most all elements of $W$), every candidate from $W$ being at distance $\ell$ from $L_1$ is dominated only by candidates from $W$. By continuing the procedure over a decreasing index of dependency layers, after at most $k^k$ steps, every path from $L_1$ to $W$ is included in $W$.

    Observe that the final $W$ has a total score at least as high as the initial optimal committee because each replacement of candidates in the procedure does not decrease the score of any voter as each voter gets at least the same amount of candidates approved per definition of the dominancy graph.

    The second claim holds, as if any optimal solution $W$ is dominated, then by making a replacement of candidates on a path $P$ from $L_1$ to $W$, such that $P \setminus W \neq \emptyset$ as done in the procedure, we produce a new solution of size $k$ with strictly higher total score which would be a contradiction to the optimality of $W$.
\end{proof}

\Cref{thm:opt-non-dominated} has several useful implications.
To solve \genthiele, it suffices to consider only the first $k$ dominancy levels.
Indeed, for any $c \in L_{k+1} \cup \dots \cup L_\delta$ and any size-$k$ committee containing $c$,
there exists a candidate $c' \in L_1 \cup \dots \cup L_k$ outside the committee that dominates $c$.

\begin{proposition}\label{prop:opt-up-to-level-k}
    There exists an optimal solution to \genthiele
    that is a subset of
    $L_1 \cup \dots \cup L_k$.
\end{proposition}

For specific Thiele rules, the instances might be even further restricted.
For example, in \ellcoverage, the score of a voter $v$ from committee $W$ equals $\min\{ |A_v \cap W|, \ell \}$.
This means that a voter can receive score at most $\ell$, which implies that only dominancy levels $L_1, \dots , L_\ell$ are of interest because removing candidates from levels $L_{\ell+1}, \dots , L_\delta$ from a solution does not change its total score.
In the case $|L_1 \cup \dots \cup L_\ell| < k$,
by taking all candidates from $L_1 \cup \dots \cup L_\ell$ to the committee, we obtain a solution with the total score equal to $\score_w(C)$ (we fill the remaining seats in the committee with arbitrary candidates).
Hence, in the case of \ellcoverage, we may assume w.l.o.g.\ that $\delta \leq \ell$.
In general, the above discussion implies the following.
\begin{proposition}
    There exists an optimal solution to \genthiele that is non-dominated and that is either a subset or a superset of
    $L_1 \cup \dots \cup L_\ell$,
    where
    $\ell = \argmax_{i \in \mathbb{N}}
    \left(
      \exists_{v \in V} (w_i^v > 0)
    \right)$.
\end{proposition}

Further structural results, specialized for the VI domain, are developed in the proof of \Cref{thm:fpt-deltac-deltav}.

\section{FPT Results on Voter Interval}\label{sec:fpt-vi}

The following result shows that \genthiele on VI is FPT parameterized by $\Delta_C + \Delta_V$. Note that \genthiele on general instances is NP-hard even if $\Delta_C = 3$ and $\Delta_V = 2$, and the hardness persists even for \pav and for \cc, which are special cases of \genthiele~\cite{AzizGGMMW15,SkowronFL16_aij_set_of_items}.
This shows that the problem is easier to solve in the VI domain assuming P$\neq$NP (a standard assumption in computational complexity theory).
This result progresses towards resolving the central open question on the complexity of PAV in the VI domain.

\begin{theorem}\label{thm:fpt-deltac-deltav}
    \genthiele on the Voter Interval domain is FPT parameterized by $\Delta_C + \Delta_V$ and is XP parameterized by $\Delta_C$.
\end{theorem}
\begin{proof}
  Due to having VI preferences with respect to $(v_1, \dots, v_n)$, for every candidate $c \in C$, we have $V_c = \{ v_i, v_{i+1}, \dots, v_j\}$ for some $1 \leq i \leq j \leq n$  (w.l.o.g., we assumed there are no candidates with an empty set of supporters).
  We denote by $\min (V_c)$ and $\max (V_c)$ the indices of the first and, respectively, the last supporting voters of candidate $c$.
  
  We order the candidates $(c_1, \dots, c_m)$ such that for every pair of indices $i < j$, one of the following conditions holds:
(1) $\max (V_{c_i}) < \max (V_{c_j})$, or
(2) $\max (V_{c_i}) = \max (V_{c_j})$ and $\min (V_{c_i}) < \min (V_{c_j})$.

  For each $i \in [n]$, we define the set $C_i = \{c \in C : \max (V_c) = i \}$, called a \emph{triangle}.
  Each triangle $C_i$ consists of candidates with the same last supporting voter $v_i$.
  Note that $\{C_1, \dots, C_n\}$ is a partition of $C$.
  The term \emph{triangle} comes from the visual shape of the approval set of candidates in $C_i$ when the approval profile is displayed as a matrix.
  In particular, after applying the above candidate ordering $(c_1, \dots, c_m)$, where candidates in $C_i$ appear consecutively, and arranging the voters according to the VI-axis,
  the approvals form triangular patterns,
  as shown in \Cref{fig:vi-triangles}.
  We denote by $\min (C_i)$ and $\max (C_i)$ the indices of the first and, respectively, the last candidate included in $C_i$.
    \begin{figure}[h]
    \centering
    \renewcommand{\arraystretch}{1.2}
    \begin{tabular}{c|ccc|cccc}
              & \( c_1 \)  & \( c_2 \)  & \( c_3 \)  & \( c_4 \)  &  \( c_5 \) &  \( c_6 \) \\
    \hline
    \( v_1 \) & \cellcheck &            &            &            &            &            \\
    \( v_2 \) & \cellcheck & \cellcheck &            & \cellcheck &            &            \\
    \( v_3 \) & \cellcheck & \cellcheck & \cellcheck & \cellcheck &            &            \\
    \( v_4 \) &            &            &            & \cellcheck & \cellcheck & \cellcheck \\
    \end{tabular}
    \caption{
      An example of a VI approval profile,
      where the voters are ordered according to the VI-axis,
      and the candidates are ordered as described in the proof of \Cref{thm:fpt-deltac-deltav}.
      In this example, there are two non-empty triangles:
      $C_3 = \{c_1, c_2, c_3\}$ and $C_4 = \{c_4, c_5, c_6\}$.
      Triangle $C_3$ is associated with voter $v_3$,
      meaning that each candidate in $C_3$ has $v_3$ as their last supporter.
      For instance, $c_1 \in C_3$ because $\max(V_{c_1}) = 3$.
      The collective set of supporters of $C_4$ is $V_{C_4} = \{v_2, v_3, v_4\}$.
      We have $\min(C_4) = 4$ and $\max(C_4) = 6$.
    }
    \label{fig:vi-triangles}
  \end{figure}
  
  Due to \Cref{thm:opt-non-dominated}, let us consider a fixed non-dominated solution $W_{\opt}$.
  It holds that $W_{\opt} \cap C_i$ is a subset of candidates from $C_i$ having indices exactly $\{ \min(C_i), \dots, j \}$ where $j = \min(C_i) + |W_{\opt} \cap C_i| - 1 \leq \max(C_i)$.
  In other words, any optimal solution contains a (possibly empty) \emph{prefix} of every triangle $C_i$.
  
  The number of all possible committees consisting of prefixes of the triangles might still be exponential in $n$.
  We solve the problem via a dynamic program. It iterates over some order of the triangles $C_1, \dots, C_n$ and considers prefixes of $C_i$ based on prefixes taken so far to a solution of a limited number (in terms of the parameters) of preceding triangles.
  A crucial observation is that candidates from triangle $C_i$ have common supporters with at most $O(\Delta_V \cdot \Delta_C)$ many candidates outside of $C_i$.
  This allows for making ``local'' decisions when considering candidates from a triangle $C_i$, where the size of locality is defined in terms of the parameters.

  More formally, let $V_{C_i}$ be the set of voters supporting any candidate from $C_i$ (hence, $|V_{C_i}|$ is the height of a triangle $C_i$).
  We observe that $|V_{C_i}| \leq \Delta_C$,
  as all members of $C_i$ have the $i$-th voter as the last supporter and the maximum number of approvals received by a candidate is $\Delta_C$. Hence, we can write $V_{C_i} \subseteq \{ \max(V_{C_i})-\Delta_C,\dots ,\max(V_{C_i})\}$.
  We have $|C_i| \leq \Delta_V$, as every candidate in $C_i$ is supported by $v_i$.
  Thus, overall, the size of every triangle is bounded in terms of the parameters of our instance.

  We base our dynamic program table
  $T[a, b, d_1, \dots, d_{|V_{C_b}|}]$ on these observations: the table saves the maximum total score of \genthiele of a (non-dominated) committee of size $a$ taken from the first $b$ triangles (that is, from $C_1 \cup \dots \cup C_b$) such that the $i$-th voter from $V_{C_b}$ approves exactly $d_i$ committee members.
  As our goal is finding a committee, we will also store (only) one committee that achieves a particular score in table $T$.
  
  We initialize the table with the first triangle, that is,
  for $b=1$, we fill exactly $\min\{|C_1|, k\}+1$ entries of $T[a, 1, d_1, \dots, d_{|V_{C_1}|}]$ by considering committees
  $C_1(a)$ being a prefix of $C_1$ of size $a$. 
  Formally,
  $C_1(0) = \emptyset$ and
  $C_1(a) = \{c_{\min(C_1)}, \dots, c_{\min(C_1) + a -1}\}$ for
  $a \in \mathbb{N}: a \leq \min\{|C_1|, k\}$.
  By $d_i(a)$ we denote the number of approved candidates by $i$-th voter from $V_{C_1}$ in a committee $C_1(a)$.
  We store the total score achieved by a committee $C_1(a)$ in a corresponding entry of $T[a, 1, d_1(a), \dots, d_{|V_{C_1}|}(a) ]$.

  When iterating over $b > 1$, we will consider committees $C_b(a)$ for $0 \leq a \leq \min\{|C_b|, k\}$ being merged with every committee stored in non-empty entries created in the previous step, i.e., when considering solution up to $(b-1)$-th triangle.
  Formally,
  we take a committee $W_\temp$ from every non-empty cell of $T[|W_\temp|, b-1, d_1^\temp, \dots, d_{|V_{C_{b-1}}|}^\temp]$
  and we create a new committee $W_\temp \cup C_b(a)$ for every $a$ such that $0 \leq a \leq \min\{|C_b|, k-|W_\temp|\}$.
  Note that in this way $|W_\temp \cup C_b(a)| \leq k$.
  Now, we compare the score of $W_\temp \cup C_b(a)$ with a score of a committee of respective entry of $T$, i.e.,
  $T[|W_\temp|+a, b, d_1(a,W_\temp), \dots, d_{|V_{C_b}|}(a,W_\temp)]$,
  where $d_i(a,W_\temp)$ is the number of approved candidates by the $i$-th voter from $V_{C_b}$ in the committee $W_\temp \cup C_b(a)$.
  If $W_\temp \cup C_b(a)$ achieves a strictly higher score, i.e.
  $ \score(W_\temp \cup C_b(a)) > T[|W_\temp|+a, b, d_1(a,W_\temp), \dots, d_{|V_{C_b}|}(a,W_\temp)]$,
  then we update the entry with this higher score, and we store $W_\temp \cup C_b(a)$ as a committee achieving this score.
  We note that the procedure is well-defined also in the case of $C_b = \emptyset$ because the only prefix of $C_b$ considered will be an empty set.
  
  To obtain the maximum total score committee,
  we search for the largest value among non-empty entries of
  $T[k, n, d_1, \dots, d_{|V_{C_n}|}]$.
  The correctness of the dynamic program follows from the fact that every such entry stores a valid committee of size~$k$, and, as argued before,
  the optimal committee~$W_\opt$ satisfies
  $W_\opt \cap C_i = C_i(|W_\opt \cap C_b|)$, i.e., $W_\opt$ consists of prefixes of the triangles $C_i$.
  This implies that, for every $b \in [n]$, the table entry
  $T[|W_\opt \cap (\cup_{i \in [b]} C_i)|, b, d_1(b), \dots, d_{|V_{C_b}|}(b)]$
  stores a committee whose total score is at least $\score_w(W_\opt \cap \cup_{i \in [b]} C_i)$
  because the relevant triangle prefixes $C_i(|W_\opt \cap C_i|)$ (or any other prefixes yielding the same score and satisfying each voter in $V_{C_b}$ by the same number of approved committee members) were explicitly considered during the construction of this entry.
  Together with the base case, i.e., filling the entries
  $T[a, 1, d_1, \dots, d_{|V_{C_1}|} ]$,
  this completes an inductive argument for the correctness of the algorithm.

  The size of $T$ is at most
  $O(k \cdot n \cdot \Delta_V^{\Delta_C})$.
  In order to fill all entries for particular $b \in [n]$, we consider at most $k+1$ prefixes of $b$-th triangle merged with every committee stored in a non-empty cell of $T$ with $b-1$.
  Therefore, we consider at most $O(k^2 \cdot \Delta_V^{\Delta_C})$ many committees, each in polynomial time.
  In total, this gives a running time of $O^*(\Delta_V^{\Delta_C})$ which is FPT with respect to $\Delta_C + \Delta_V$. 
\end{proof}
A consequence from the above presented proof is an FPT algorithm parameterized by $\Delta_C + k$.
That is, the indices of $T$, which encode how many times every voter involved is represented, never exceed the committee size $k$.
\begin{corollary}\label{cor:vi-k-deltac}
    \genthiele on the Voter Interval domain can be solved in $O^*(k^{\Delta_C})$ time.
\end{corollary}
We note that \Cref{thm:fpt_k_deltac} provides an FPT algorithm parameterized by $\Delta_C + k$ for every instance,
but its running time is double-exponential on $k$ in contrast to the non-exponential dependence on $k$ in the case of VI (\Cref{cor:vi-k-deltac}).
It is another example of parameterization for which much more efficient algorithms exist for the VI structured domain.

\section{FPT Results on General Instances}\label{sec:general-instances}

In the next two subsections, we present two FPT algorithms for \genthiele which answer affirmatively two open questions known in the literature~\cite{YangW18-aamas18,YangW23-jaamas},
which were asked for a special case of PAV.
In \Cref{sec:poly-deltac-2} we provide a polynomial-time algorithm for instances with $\Delta_C = 2$.
In \Cref{sec:fpt-d} we give an FPT algorithm parameterized by the total score of an optimal committee.

\subsection[Polynomial-Time Algorithm when Each Candidate is Approved by at Most Two Voters]{Polynomial-Time Algorithm when $\Delta_C = 2$}\label{sec:poly-deltac-2}
Our polynomial-time algorithm for \genthiele with $\Delta_C = 2$ is based on a generalization of an integer linear program (ILP) studied by Peters~\cite{Peters18_aaai}.
The ILP formulation in \cite{Peters18_aaai} is defined for any \wthiele, but it can be easily adjusted to \genthiele by modifying the objective function~\cite{SornatWX22_ijcai}.
For a given election $(C,V,A,k)$ the ILP for \genthiele \eqref{ilp_gt} is defined as follows, where roughly speaking, the $y_c$ variables mark selected candidates, the $x_{v,i}$ variables track voter satisfaction levels and the objective ensures that the highest available values in a Thiele sequence are always chosen:
\begin{align}
    \tag{ILP-GT}
    \textrm{maximize}\quad  \sum_{v \in V} \sum_{i \in [k]} w^v_i   &\cdot x_{v,i}                   \label{ilp_gt}\\
    \textrm{subject to}\quad\quad\quad\hspace{3pt} \sum_{c \in C}                  y_c     &= k  \label{ilp_gt_cardinality} \\
                        \sum_{i \in [k]}                x_{v,i} &\leq \sum_{c \in A_v} y_c       \quad\quad\;\forall v \in V \label{ilp_gt_counting}\\
                        x_{v,i} &\in \{0,1\} \quad\quad\quad\forall v \in V, i \in [k] \nonumber\\
                        y_c     &\in \{0,1\} \quad\quad\quad\hspace{1pt}\forall c \in C \nonumber
\end{align}

Peters~\cite{Peters18_aaai} showed that the constraint matrix of \ref{ilp_gt} is totally unimodular (TU) when an approval profile is CI, hence an optimal solution can be found in polynomial time.
He actually argued that if an approval profile $A$ captured as a matrix, one can find an optimal solution in polynomial time even when an additional row with all-$1$s is appended (corresponding to Constraint~\eqref{ilp_gt_cardinality}), as the resulting matrix is still TU.
In the case of $A$ being CI, total unimodularity was an immediate implication from the fact that $A$ being CI has consecutive one property, and because the additional row consists of only $1$s (hence, it is consistent with the consecutive $1$s property).

In the case where $A$ is VI, total unimodularity is not necessarily preserved.
While the transpose $A^T$ satisfies the consecutive $1$s property and is thus TU,
this property may be lost when appending an all-ones row to $A$.
Specifically, the matrix
$\matrixapp{\mathbf{1}_m}{A}$ obtained by adding a row of $1$s of length $m$ to $A$---does not, in general, have the consecutive $1$s property in its transpose.
There exists a small VI approval profile that yields a non-TU matrix; see Example~\eqref{example:3x4},
where a VI profile with $4$ candidates and $3$ voters together with the cardinality constraint (the first row) is not TU (its determinant is $-2$).
\begin{equation}
  \label{example:3x4}
    \begin{bmatrix}
            1 & 1 & 1 & 1 \\
            1 & 1 & 0 & 0 \\
            1 & 0 & 1 & 0 \\
            1 & 0 & 0 & 1
    \end{bmatrix}
\end{equation}

Example \eqref{example:3x4} extends to Example~\eqref{example:2m} with $2m$ candidates and $2m-1$ voters for any $m \geq 2$, whose determinant is $2-2m \leq -2$.
Hence, these matrices are not TU either.
\begin{equation}
\label{example:2m}
  \begin{bmatrix}
     1                     & \mathbf{1}_{2m-1} \\
     (\mathbf{1}_{2m-1})^T & I_{2m-1}
  \end{bmatrix}
\end{equation}

However, there are more classes of integer programs that can be solved in polynomial time. In particular, if the coefficients of the constraint matrix are in $\{-2,-1,0,1,2\}$ and the sum of absolute values is at most $2$ for each column, then the problem is polynomial time solvable~\cite{schrijver2003combinatorial}. Such matrices are called \emph{generalized matching} matrices. 

\begin{theorem}[Schrijver~\cite{schrijver2003combinatorial}]\label{thm:matchingMatrix}
ILPs with a generalized matching matrix can be solved in strongly polynomial time.
\end{theorem}

These matrices correspond to problems called (generalized) matching problems, hence the name. The corresponding ILPs capture a variety of well-known problems in polynomial time such as minimum cost flow, minimum cost ($b$-)matching and certain graph factor problems~\cite{schrijver2003combinatorial}. These structures even remain FPT time solvable parameterized by the number $p$ of additional columns~\cite{lassota2025parameterizedalgorithmsmatchinginteger}.
However, we are facing the problem of having an additional row, which is shown to be, in general, not FPT assuming FPT is not equal to W[1] (a common hypothesis in parameterized complexity similar to P versus NP)~\cite{lassota2025parameterizedalgorithmsmatchinginteger}.
Fortunately, we can alter \ref{ilp_gt} slightly and obtain a generalized matching matrix. 

Before we show the proof, note that if there is a candidate that is not approved by any voter, we can delete this candidate from the instance. If $k$ is larger than the remaining candidates, we take all those remaining candidates and greedily fill up our committee with any candidate from the discarded list to obtain an optimal solution. If there is at least one candidate that is approved by only one voter, then we add a new dummy voter $v_d$ that approves exactly all of candidates with only one approval. This voter $v_d$ will have a zero contribution to the objective function. Hence, in the following, we assume w.l.o.g.\ that all candidates are approved by exactly two voters.

\begin{theorem}\label{thm_matrix_tu}
    \genthiele with $\Delta_C = 2$ can be solved in polynomial time.
\end{theorem}
\begin{proof}
We modify the formulation of \ref{ilp_gt} as follows.
First, we replace Constraint~\eqref{ilp_gt_cardinality} by the following equality:
\begin{align}\label{constraint_x_2k}
    \sum_{v \in V} \sum_{i \in [k]} x_{v,i} = 2k .
\end{align}
Second, we strengthen Constraint~\eqref{ilp_gt_counting} by replacing the inequality with equality:
\begin{align}\label{constraint_x_y}
     \sum_{i \in [k]} x_{v,i} = \sum_{c \in A_v} y_c \quad\quad\;\forall v \in V .
\end{align}

Constraint~\eqref{ilp_gt_cardinality} previously ensured that exactly $k$ candidates are selected.
As argued above, each candidate has exactly two supporters, so choosing any candidate (represented by variable $y_c$) implies that two variables $x_{v,i}$ must be set to $1$ (one for each approving voter) for valid solutions.
Thus, replacing \eqref{ilp_gt_cardinality} by \eqref{constraint_x_2k} and enforcing equality in Constraint~\eqref{constraint_x_y} shifts the responsibility for enforcing the committee size from the $y_c$ variables to the $x_{v,i}$ variables.

After these modifications, each variable $x_{v,i}$ and $y_c$ (each corresponding to one column of the constraint matrix) appears in exactly 2 constraints with coefficients in $\{-1,1\}$ (and $0$ everywhere else).
This implies that the resulting constraint matrix is a generalized matching matrix.
By applying Theorem~\ref{thm:matchingMatrix}, we obtain a solution in polynomial time.
\end{proof}

\subsection{FPT Algorithm Parameterized by the Score}\label{sec:fpt-d}

Next, we prove that \genthiele is FPT parameterized by $k+\Delta_C$.
The result relies on reducing the problem to a set cover variant called $p$-partial set cover that is solvable efficiently. 

In the $p$-partial set cover problem, we are given a universe $U$ of $t$ elements, and a set of subsets $\mathcal{S}$ of size $s$. The goal is to cover at least $p$ different elements of $U$ using the minimal number of sets in $\mathcal{S}$. This problem is FPT time solvable parameterized by $p$. 

\begin{theorem}[Bl{\"{a}}ser~\cite{Blaser03}]\label{thm:pSetCover}
A minimum $p$-partial set cover can be computed in time $2^{O(p)}\cdot s \cdot t$.
\end{theorem}

Note that this algorithm also solves the weighted $p$-partial set cover problem where each set in $\mathcal{S}$ is assigned a weight and the goal is to find the minimum weight set cover hitting at least $p$ distinct elements of $U$.

Our algorithm reduces \genthiele to $p$-partial set cover with $p$ bounded by a function of $k$ and $\Delta_C$.
We view each voter as an element of the universe and every candidate as the set of supporters.
Contrary to the set cover problem, in \genthiele, it can be beneficial to ``cover'' a voter multiple times. To implement this in our reduction, we use color coding, a powerful tool to design FPT time algorithms (see, e.g., the book of Cygan \etal~\cite{CyganFKLMPPS15_fpt_book}). Color coding is used to color solution candidates such that, with high probability, an optimal solution only takes one solution candidate per color. We use the colors to color all approvals where each color indicates the additive contribution to the valuation function. This allows us to create different elements not just for every approval, but also the way it contributes to the objective function, e.g., the number of times a voter has been ``covered''.

To derandomize color coding, splitters have been introduced~\cite{NaorSS95}. Splitters compute a number of colorings instead of a single one such that at least one coloring has the desired property that an optimal solution only takes one solution candidate per color. This can be expressed in terms of hash functions: 

\begin{definition}
An $(n,k,\ell)$ splitter is a family of hash functions $F$ from $\{1, 2, \dots, n\}$ to $\{1, 2, \dots, \ell\}$ such that for every $S \subseteq \{1,2,\dots,n\}$ with $|S|=k$, there exists a function $f \in F$ that splits $S$ evenly; that is, for every $j,j'\leq \ell$, we have $|f^{-1}(j)\cap S|$ and $|f^{-1}(j')\cap S|$ differ by at most 1. 
\end{definition}

\begin{lemma}[Naor, Schulman, and Srinivasan~\cite{NaorSS95}]\label{lem:splitters}
There exists an $(n,k,k)$ splitter of size $e^kk^{O(\log(k))}\log(n)$ which is computable in time $e^kk^{O(\log(k))} \cdot n\log(n)$.
\end{lemma}

Equipped with splitters, we can finally present the main result of this section.

\begin{theorem}\label{thm:fpt_k_deltac}
    \genthiele is FPT parameterized by $k+\Delta_C$.
\end{theorem}
\begin{proof}
    The main idea of this proof is to reduce \genthiele to $p$-partial set cover such that $p$ is bounded by $k$ and $\Delta_C$. Observe that any $k$ candidates can satisfy at most $k \cdot \Delta_C$ different voters. As the function is bounded by the parameter, we can guess the correct coverage. This already resembles the $p$-partial set cover problem where voters are elements, candidates are sets of voters that approve them. However, it does not capture that is adds to the objective function if we select two candidates that cover the same voters. We simulate this by using color coding to color all approvals. Each color indicates the additive contribution to the valuation function. This allows us create different elements not just for every approval, but also the way it contributes to the objective function. In particular, the additive contribution to the objective function is mimnicked by introducing different quantities of covered elements.  
 
    Formally, we are given candidates $C=\{c_1, c_2, \dots, c_m\}$ and voters $V=\{v_1, v_2, \dots, v_n\}$. We visualize the approvals of voters to candidates as a matrix $A\in \{0,1\}^{n\times m}$. An entry $a_{i,j}=1$ if voter $v_i$ approves candidate $c_j$, otherwise $a_{i,j}=0$. Denote by $q$ the number of $1$s in the matrix. Observe that $q \leq m \cdot \Delta_C$.

    We color all $1$s in the matrix with $k\Delta_C$ colors $1, 2, \dots, k\Delta_C$. Each consecutive $\Delta_C$ colors belong to the same type, e.g., colors $1+\Delta_C(i-1)$ to $\Delta_C \cdot i$ belong to class type $i$. Color class $c$ states that the corresponding $1$s contribute $1/c$ to the objective function. As there are at most $k\Delta_C$ approvals in any solution with $k$ candidates, we can compute an $(q,k\Delta,k\Delta)$ splitter that contains are a coloring that colors all $1$s of an optimal solution (see Lemma~\ref{lem:splitters}) differently. To guarantee the right meaning of the color, we further try all permutations of color classes for each coloring. 

   Assume we are given a correct coloring, that is, an optimal solution is colored correctly. We use this coloring and the approval matrix $A$ to compute a matrix $A' \in \{0,1\}^{2k!n \times n}$ that splits voters according to the color of their approvals. Each column $a_{\ast, j}$ of $A$ defines column $a'_{\ast, j}$ of $A'$: the entries $1+2k!(i-1)$ to $2k!(i)$ entries of $a'_{\ast, j}$ belong to the $i$th entry of $a'_{\ast, j}$: If $a'_{\ast, j} = 0$, all $2k!$ entries are 0. Otherwise, if $a'_{\ast, j}=1$ and $1$ is of color class $c$, then we put a $1$ in all entries from $1+\sum_{\ell=2}^c k!/c$ to $\sum_{\ell=2}^c k!/c + k!/c$, all other entries are 0. 

    We can now describe the $p$-partial set cover instance. Every row in $A'$ corresponds to an element, hence $U$ has size $2k!n$. Every column $a'_{\ast, j}$ corresponds to a set $S_j$, where a $1$ in the corresponding row indicates that the element is in $S_j$. Hence, there are $m$ many sets in $\mathcal{S}$. 
    
    Regarding $q$, we cover at least $k$ many elements in any solution of size $k$ (if there is a candidate that has zero approvals, we can delete it in a preprocessing step), and at most $k\cdot k!\Delta_C$. Using binary search for the largest value within these bounds that does not use more than $k$ elements in a solution, we directly get our optimal solution by taking the corresponding candidates.

    \textit{Correctness.} If there exists an optimal solution of value OPT to \genthiele, then there exists an optimal solution to the $p$-partial set cover instance of value OPT $\cdot 2k!$. Let $C^{\textrm{OPT}} \subseteq C$ be the committee of an optimal solution. Assume that we are given a coloring (which we are given with high probability) such that all $1$s in $A$ are colored correctly, e.g., each row with $\ell$ approvals of the submatrix $A^{\textrm{OPT}}$ restricted to the columns corresponding to $C^{\textrm{OPT}}$ uses each colorclass $1, 2, \dots, k$ exactly once. Then one can see that the set $S^{\textrm{OPT}}$ corresponding to $C^{\textrm{OPT}}$ will give the desired solution.

    Now assume that we are given an optimal solution over all colorings of value OPT for $p$-partial cover that uses at most $k$ sets $S^{\textrm{OPT}}$. In the submatrix $A'^{\textrm{OPT}}$ restricted to the sets $\mathcal{S}^{\textrm{OPT}}$ of the optimal solution, each row must use  each color $1, 2, \dots, k$ exactly once as otherwise there would be a better coloring achieving a higher value for any even for the same set. Now, this corresponds to the optimal solution $C^{\textrm{OPT}}$ with value OPT$/2k!$. If there would be a better solution to \genthiele, it would immediately imply a better solution to the $p$-partial set cover instance as outlined above, which is a contradiction. 
    
    \textit{Running time.} Coloring all $1$s takes time $O(mn)$. Computing $A'$ takes time $2k!nm$. Computing the splitter takes time $e^kk^{O(\log(k))}q\log(q)$. Using Theorem~\ref{thm:pSetCover}, it takes time $2^{O(k\cdot k!\Delta_C)}\cdot 2k!n \cdot m$ for each of the $e^kk^{O(\log(k))}\log(n)$ guesses and $(\Delta_C k)^k$ permutations of the color classes.
    This gives an overall running time of
    \[
        2^{O(k^4\cdot k!\Delta_C \log(k!\Delta_C))} \cdot n \cdot m^2.\qedhere
    \]
\end{proof}

This yields the following FPT algorithm.

\begin{proposition}\label{prop:thiele_fpt_d}
   For every Thiele sequence $w$, \wthiele can be solved in time $2^{d^{O(d)}} \cdot nm^{O(1)}$, hence it is FPT parameterized by $d$, the total score of an optimal solution.
\end{proposition}
\begin{proof}
    In the case $n < k$, we can find a committee that has at least one representative of every voter (e.g., by using a greedy algorithm).
    Since $w_1 = 1$, we have $n \leq d$.
    As \wthiele is FPT with respect to~$n$ \cite{FaliszewskiSST18_scw}, in this case their algorithm is also FPT with respect to~$d$ with the running time at most $O^*(2^{2^{O(d)}}) \leq 2^{2^{O(d)}} \cdot (dm)^{O(1)}$.
    
    In the case $n \geq k$, using a greedy algorithm, we can find a committee that represents at least $k$ voters (every taken candidate covers at least one additional voter or it already represents $k$ voters).
    Since $w_1 = 1$, we have $k \leq d$.
    By taking candidate $c$ to the committee with the highest $|V_c|$, it has to be $\Delta_C \leq d$.
    Therefore, we have $k+\Delta_C \leq 2d$ and an FPT algorithm parameterized by $k+\Delta_C$ is an FPT algorithm parameterized by $d$.
    Hence, the application of \Cref{thm:fpt_k_deltac} provides an algorithm with the running time at most
    $$
    2^{O(d^5\cdot d! \log(d!d))} \cdot n \cdot m^2.
    $$
    
    The running time from both cases can be upper-bounded by
    $2^{d^{O(d)}} \cdot n \cdot m^{O(1)}$,
    which has only linear time dependency on $n$.
    This finishes the proof.
\end{proof}

As discussed in the Related Work section,
a recent result by Gupta, Jain, Saha, Saurabh and Upasana~\cite{GuptaJSSU2025_independent_fpt} also provides an FPT algorithm parameterized by $d$.
They solved Submodular Multiwinner Election with parameter $t$ that is the number of approvals in a solution. Note that $\Delta_C \leq t \leq \Delta_C \cdot k$ and that $t$ can be as large as $\Delta_C \cdot k$. Their results also use color-coding. In particular, they color all candidates with $k$ colors, and all voters with $t$ colors, i.e., every candidates of an optimal solution and all of the corresponding, at most $t$ approvals in an optimal solution are colored with a different color. It remains to now greedily take the candidates from each color that indeed covers all approvals with the desired colors. Their running time is $k^kt^{t+1}k!t!n^{O(1)}\cdot m^O(1)$ which for, e.g. for \pav, this yields $2^{2(2t)^2}\cdot t^{2t+1} \cdot n^{O(1)}\cdot m^{O(1)}$. 
In comparison to our algorithm, our running time in terms of $\Delta_C$ and $n$ are better, in particular the dependency on $\Delta_C$ is bounded by a factor of $2^{\Delta_C\log(\Delta_C)}$ and our algorithm runs in time truly linear in $n$. We would also like to highlight that while both algorithms use color coding at its heart, the approaches differ significantly. While they color candidates and voters, we just color the approvals of each voter. This results in different amount, kind, and meaning of guesses and construction of the overall solution.

\section{Conclusion and Future Work}

We presented new algorithms for computing optimal committees under Thiele rules.
We identified structural properties of optimal solutions,
which on Voter Interval instances enable a dynamic programming approach over a chain of subsets of candidates.
We also resolved an open problem by showing that winner determination under any Thiele rule is polynomial-time solvable when each candidate is approved by at most two voters, using an ILP-based approach.
Furthermore, we provided an FPT algorithm parameterized by $k+\Delta_C$ using color-coding technique, which we apply to obtain an FPT algorithm parameterized by the total score $d$.

One question raised in the conference version of this paper~\cite{LassotaS2026_pavvi} was whether winner determination under PAV on Voter Interval instances is polynomial-time solvable or NP-hard. This question has since been resolved by two independent groups, each of which obtained a polynomial-time algorithm~\cite{ManurangsiS26_thiele_vi_arxiv,AvramidisLSV26_arxiv}.

Nevertheless, several directions remain open. In particular, it remains unknown whether the ILP-based result for $\Delta_C = 2$ can be replaced by a purely combinatorial algorithm---a question that also arises in the context of Thiele rules on Candidate Interval profiles~\cite{Peters18_aaai}.

It is natural to ask whether \Cref{thm:fpt-deltac-deltav} extends to more general preference domains such as
Voter-Candidate Interval \cite{DongBWBE25_aamas25,ElkindFIMSS24_teac24,GodziszewskiB0F21_vci_aaai21}.
In contrast to \cite{DongBWBE25_aamas25},
whose objective allows discarding dominated candidates,
dominated candidates crucially affect Thiele scores.

Finally, our structural results and algorithms could be used to reason about tied committees,
for example by analyzing possible and necessary winners under Thiele rules,
since our methods that compute optimal scores and can handle preselected candidates (by adjusting Thiele sequences in \genthiele).

\section*{Acknowledgements}
We thank Andrei Constantinescu and Piotr Faliszewski for discussions of our results that helped us clarify and improve their presentation.
We also thank the anonymous reviewers for their valuable feedback.

Alexandra Lassota was supported by the Dutch Research Council (NWO) under project number VI.Veni.242.293.
Krzysztof Sornat was supported by the European Research Council (ERC) under the European Union’s Horizon 2020 research and innovation programme (grant agreement No 101002854).

\bibliographystyle{alpha}
\bibliography{bib}

@article{AvramidisLSV26_arxiv,
  author       = {Dimitris Avramidis and
                  Alexandra Lassota and
                  Ulrike Schmidt{-}Kraepelin and
                  Adrian Vetta},
  title        = {Computing {T}hiele Rules on Interval Elections and their Generalizations},
  journal      = {CoRR},
  volume       = {abs/2605.03067},
  year         = {2026}
}

@inproceedings{AzizGGMMW15,
  author       = {Haris Aziz and
                  Serge Gaspers and
                  Joachim Gudmundsson and
                  Simon Mackenzie and
                  Nicholas Mattei and
                  Toby Walsh},
  title        = {Computational Aspects of Multi-Winner Approval Voting},
  booktitle    = {Proceedings of the 2015 International Conference on Autonomous Agents and Multiagent Systems ({AAMAS} 2015)},
  pages        = {107--115},
  year         = {2015}
}

@inproceedings{BarmanFF21_concave_coverage,
  author       = {Siddharth Barman and
                  Omar Fawzi and
                  Paul Ferm{\'{e}}},
  title        = {Tight Approximation Guarantees for Concave Coverage Problems},
  booktitle    = {Proceedings of the 38th International Symposium on Theoretical Aspects of Computer Science ({STACS} 2021)},
  pages        = {9:1--9:17},
  year         = {2021}
}

@article{BarmanFGG22_lcoverage,
  author       = {Siddharth Barman and
                  Omar Fawzi and
                  Suprovat Ghoshal and
                  Emirhan G{\"{u}}rpinar},
  title        = {Tight Approximation Bounds for Maximum Multi-Coverage},
  journal      = {Math. Program.},
  volume       = {192},
  number       = {1},
  pages        = {443--476},
  year         = {2022}
}

@article{BetzlerSU13_jair,
  author       = {Nadja Betzler and
                  Arkadii Slinko and
                  Johannes Uhlmann},
  title        = {On the Computation of Fully Proportional Representation},
  journal      = {J. Artif. Intell. Res.},
  volume       = {47},
  pages        = {475--519},
  year         = {2013}
}

@article{Blaser03,
  author       = {Markus Bl{\"{a}}ser},
  title        = {Computing Small Partial Coverings},
  journal      = {Inf. Process. Lett.},
  volume       = {85},
  number       = {6},
  pages        = {327--331},
  year         = {2003}
}

@inproceedings{BoehmerBCG0S24_polkadot,
  author       = {Niclas Boehmer and
                  Markus Brill and
                  Alfonso Cevallos and
                  Jonas Gehrlein and
                  Luis S{\'{a}}nchez Fern{\'{a}}ndez and
                  Ulrike Schmidt{-}Kraepelin},
  title        = {Approval-Based Committee Voting in Practice: {A} Case Study of (over-)Representation in the Polkadot Blockchain},
  booktitle    = {Proceedings of the 38th {AAAI} Conference on Artificial Intelligence ({AAAI} 2024)},
  pages        = {9519--9527},
  year         = {2024}
}

@inproceedings{BredereckF0KN20_aaai,
  author       = {Robert Bredereck and
                  Piotr Faliszewski and
                  Andrzej Kaczmarczyk and
                  Dusan Knop and
                  Rolf Niedermeier},
  title        = {Parameterized Algorithms for Finding a Collective Set of Items},
  booktitle    = {Proceedings of the Thirty-Fourth {AAAI} Conference on Artificial Intelligence ({AAAI} 2020)},
  pages        = {1838--1845},
  year         = {2020}
}

@inproceedings{BredereckFNST15_adt,
  author       = {Robert Bredereck and
                  Piotr Faliszewski and
                  Rolf Niedermeier and
                  Piotr Skowron and
                  Nimrod Talmon},
  title        = {Elections with Few Candidates: {P}rices, Weights, and Covering Problems},
  booktitle    = {Proceedings of the 4th Conference on Algorithmic Decision Theory ({ADT} 2015)},
  pages        = {414--431},
  year         = {2015}
}

@article{BredereckFNST20_tcs,
  author       = {Robert Bredereck and
                  Piotr Faliszewski and
                  Rolf Niedermeier and
                  Piotr Skowron and
                  Nimrod Talmon},
  title        = {Mixed Integer Programming with Convex/Concave Constraints:
                  {F}ixed-Parameter Tractability and Applications to Multicovering and Voting},
  journal      = {Theor. Comput. Sci.},
  volume       = {814},
  pages        = {86--105},
  year         = {2020}
}

@inproceedings{BrillP23_ejrp,
  author       = {Markus Brill and
                  Jannik Peters},
  title        = {Robust and Verifiable Proportionality Axioms for Multiwinner Voting},
  booktitle    = {Proceedings of the 24th {ACM} Conference on Economics and Computation ({EC} 2023)},
  pages        = {301},
  year         = {2023}
}

@inproceedings{ByrkaSS18_hkm,
  author       = {Jaros\l{}aw Byrka and
                  Piotr Skowron and
                  Krzysztof Sornat},
  title        = {Proportional Approval Voting, Harmonic k-Median, and Negative Association},
  booktitle    = {Proceedings of the 45th International Colloquium on Automata, Languages, and Programming ({ICALP} 2018)},
  pages        = {26:1--26:14},
  year         = {2018}
}

@article{ChamberlinCourant83,
  author       = {John R. Chamberlin and
                  Paul N. Courant},
  title        = {Representative Deliberations and Representative Decisions: {P}roportional Representation and the {B}orda Rule},
  journal      = {Am. Political Sci. Rev.},
  volume       = {77},
  issue        = {03},
  month        = {9},
  year         = {1983},
  pages        = {718--733}
}

@inproceedings{ConstantinescuE21_sc_linear,
  author       = {Andrei Costin Constantinescu and
                  Edith Elkind},
  title        = {Proportional Representation under Single-Crossing Preferences Revisited},
  booktitle    = {Proceedings of the 35th {AAAI} Conference on Artificial Intelligence ({AAAI} 2021)},
  pages        = {5286--5293},
  year         = {2021}
}

@book{CyganFKLMPPS15_fpt_book,
  author       = {Marek Cygan and
                  Fedor V. Fomin and
                  {\L{}}ukasz Kowalik and
                  Daniel Lokshtanov and
                  D{\'{a}}niel Marx and
                  Marcin Pilipczuk and
                  Micha\l{} Pilipczuk and
                  Saket Saurabh},
  title        = {Parameterized Algorithms},
  publisher    = {Springer},
  year         = {2015}
}

@inproceedings{DongBWBE25_aamas25,
  author       = {Chris Dong and
                  Martin Bullinger and
                  Tomasz W\k{a}s and
                  Larry Birnbaum and
                  Edith Elkind},
  title        = {Selecting Interlacing Committees},
  booktitle    = {Proceedings of the 24th International Conference on Autonomous Agents and Multiagent Systems ({AAMAS} 2025)},
  pages        = {630--638},
  year         = {2025}
}

@inproceedings{DudyczMMS20_ijcai,
  author       = {Szymon Dudycz and
                  Pasin Manurangsi and
                  Jan Marcinkowski and
                  Krzysztof Sornat},
  title        = {Tight Approximation for Proportional Approval Voting},
  booktitle    = {Proceedings of the 29th International Joint Conference on Artificial Intelligence ({IJCAI} 2020)},
  pages        = {276--282},
  year         = {2020}
}

@article{ElkindFIMSS24_teac24,
  author       = {Edith Elkind and
                  Piotr Faliszewski and
                  Ayumi Igarashi and
                  Pasin Manurangsi and
                  Ulrike {Schmidt{-}Kraepelin} and
                  Warut Suksompong},
  title        = {The Price of Justified Representation},
  journal      = {{ACM} Trans. Economics and Comput.},
  volume       = {12},
  number       = {3},
  pages        = {11:1--11:27},
  year         = {2024}
}

@inproceedings{ElkindL15_ijcai15,
  author       = {Edith Elkind and
                  Martin Lackner},
  title        = {Structure in Dichotomous Preferences},
  booktitle    = {Proceedings of the 24th International Joint Conference on Artificial Intelligence ({IJCAI} 2015)},
  pages        = {2019--2025},
  year         = {2015}
}

@incollection{ElkindLP2017trends_book,
  author        = {Edith Elkind and
                   Martin Lackner and
                   Dominik Peters},
  title         = {Structured Preferences},
  editor        = {Ulle Endriss},
  booktitle     = {Trends in Computational Social Choice},
  chapter       = {10},
  pages         = {187--207},
  publisher     = {AI Access},
  year          = {2017}
}

@article{ElkindLP25_arxiv_preference_restrictions,
  author       = {Edith Elkind and
                  Martin Lackner and
                  Dominik Peters},
  title        = {Preference Restrictions in Computational Social Choice: {A} Survey},
  journal      = {CoRR},
  volume       = {abs/2205.09092v2},
  year         = {2025}
}

@article{FaliszewskiHHR11,
  author       = {Piotr Faliszewski and
                  Edith Hemaspaandra and
                  Lane A. Hemaspaandra and
                  J{\"{o}}rg Rothe},
  title        = {The Shield that Never was: {S}ocieties with Single-Peaked Preferences are More Open to Manipulation and Control},
  journal      = {Inf. Comput.},
  volume       = {209},
  number       = {2},
  pages        = {89--107},
  year         = {2011}
}

@article{FaliszewskiSST18_scw,
  author       = {Piotr Faliszewski and
                  Piotr Skowron and
                  Arkadii Slinko and
                  Nimrod Talmon},
  title        = {Multiwinner Analogues of the Plurality Rule: {A}xiomatic and Algorithmic Perspectives},
  journal      = {Soc. Choice Welf.},
  volume       = {51},
  number       = {3},
  pages        = {513--550},
  year         = {2018}
}

@inproceedings{GodziszewskiB0F21_vci_aaai21,
  author       = {Micha\l{} Tomasz Godziszewski and
                  Pawe\l{} Batko and
                  Piotr Skowron and
                  Piotr Faliszewski},
  title        = {An Analysis of Approval-Based Committee Rules for {2D}-{E}uclidean Elections},
  booktitle    = {Proceedings of the 35th {AAAI} Conference on Artificial Intelligence ({AAAI} 2021)},
  pages        = {5448--5455},
  year         = {2021}
}

@inproceedings{GuptaJSSU2025_independent_fpt,
  author       = {Sushmita Gupta and
                  Pallavi Jain and
                  Souvik Saha and
                  Saket Saurabh and
                  Anannya Upasana},
  title        = {More Efforts Towards Fixed-Parameter Approximability of Multiwinner Rules},
  booktitle    = {Proceedings of the 34th International Joint Conference on Artificial Intelligence ({IJCAI} 2025)},
  pages        = {3891--3899},
  year         = {2025}
}

@inproceedings{JainST20_ijcai,
  author       = {Pallavi Jain and
                  Krzysztof Sornat and
                  Nimrod Talmon},
  title        = {Participatory Budgeting with Project Interactions},
  booktitle    = {Proceedings of the 29th International Joint Conference on Artificial Intelligence ({IJCAI} 2020)},
  pages        = {386--392},
  year         = {2020}
}

@article{LacknerS21,
  author       = {Martin Lackner and
                  Piotr Skowron},
  title        = {Consistent Approval-based Multi-winner Rules},
  journal      = {J. Econ. Theory},
  volume       = {192},
  pages        = {105173},
  year         = {2021}
}

@book{lackner23abc_book,
  author       = {Martin Lackner and
                  Piotr Skowron},
  title        = {Multi-Winner Voting with Approval Preferences},
  publisher    = {Springer},
  year         = {2023}
}

@inproceedings{lassota2025parameterizedalgorithmsmatchinginteger,
  author       = {Alexandra Lassota and
                  Koen Ligthart},
  title        = {Parameterized Algorithms for Matching Integer Programs with Additional Rows and Columns},
  booktitle    = {Proceedings of the 52nd International Colloquium on Automata, Languages, and Programming ({ICALP} 2025)},
  pages        = {112:1--112:18},
  year         = {2025}
}

@inproceedings{LassotaS2026_pavvi,
  author       = {Alexandra Lassota and
                  Krzysztof Sornat},
  title        = {Algorithms for Structured Elections under {T}hiele Voting Rules},
  booktitle    = {Proceedings of the 40th {AAAI} Conference on Artificial Intelligence ({AAAI} 2026)},
  pages        = {17084--17092},
  year         = {2026}
}

@inproceedings{LiuG16_mav_vi,
  author       = {Hong Liu and
                  Jiong Guo},
  title        = {Parameterized Complexity of Winner Determination in Minimax Committee Elections},
  booktitle    = {Proceedings of the 2016 International Conference on Autonomous Agents and Multiagent Systems ({AAMAS} 2016)},
  pages        = {341--349},
  year         = {2016}
}

@article{ManurangsiS26_thiele_vi_arxiv,
  author       = {Pasin Manurangsi and
                  Krzysztof Sornat},
  title        = {Polynomial-Time Algorithm for {T}hiele Voting Rules with Voter Interval
                  Preferences},
  journal      = {CoRR},
  volume       = {abs/2604.05953},
  year         = {2026}
}

@inproceedings{NaorSS95,
  author       = {Moni Naor and
                  Leonard J. Schulman and
                  Aravind Srinivasan},
  title        = {Splitters and Near-Optimal Derandomization},
  booktitle    = {Proceedings of the 36th Annual Symposium on Foundations of Computer Science ({FOCS} 1995)},
  pages        = {182--191},
  year         = {1995}
}

@inproceedings{Peters18_aaai,
  author       = {Dominik Peters},
  title        = {Single-Peakedness and Total Unimodularity:
                  {N}ew Polynomial-Time Algorithms for Multi-Winner Elections},
  booktitle    = {Proceedings of the 32nd {AAAI} Conference on Artificial Intelligence ({AAAI} 2018)},
  pages        = {1169--1176},
  year         = {2018}
}

@article{PetersL20_spoc_jair,
  author       = {Dominik Peters and
                  Martin Lackner},
  title        = {Preferences Single-Peaked on a Circle},
  journal      = {J. Artif. Intell. Res.},
  volume       = {68},
  pages        = {463--502},
  year         = {2020}
}

@book{schrijver2003combinatorial,
  author       = {Alexander Schrijver},
  title        = {Combinatorial Optimization: {P}olyhedra and Efficiency},
  series       = {Algorithms and Combinatorics},
  volume       = {24},
  publisher    = {Springer},
  year         = {2003}
}

@article{SkowronFL16_aij_set_of_items,
  author       = {Piotr Skowron and
                  Piotr Faliszewski and
                  J{\'{e}}r{\^{o}}me Lang},
  title        = {Finding a Collective Set of Items: {F}rom Proportional Multirepresentation to Group Recommendation},
  journal      = {Artif. Intell.},
  volume       = {241},
  pages        = {191--216},
  year         = {2016}
}

@inproceedings{SornatWX22_ijcai,
  author       = {Krzysztof Sornat and
                  Virginia {Vassilevska Williams} and
                  Yinzhan Xu},
  title        = {Near-Tight Algorithms for the {C}hamberlin-{C}ourant and {T}hiele Voting Rules},
  booktitle    = {Proceedings of the 31st International Joint Conference on Artificial Intelligence ({IJCAI} 2022)},
  pages        = {482--488},
  year         = {2022}
}

@incollection{Thiele95,
  author    = {Thorvald Thiele},
  booktitle = {Oversigt over det Kongelige Danske Videnskabernes Selskabs Forhandlinger (in {D}anish)},
  pages     = {415--441},
  title     = {Om Flerfoldsvalg},
  year      = {1895},
  publisher = {K{\o}benhavn: A.F. H{\o}st}
}

@article{YangW23-jaamas,
  author       = {Yongjie Yang and
                  Jianxin Wang},
  title        = {Parameterized Complexity of Multiwinner Determination:
                  {M}ore Effort Towards Fixed-Parameter Tractability},
  journal      = {Auton. Agents Multi Agent Syst.},
  volume       = {37},
  number       = {2},
  pages        = {28},
  year         = {2023}
}

@inproceedings{YangW18-aamas18,
  author       = {Yongjie Yang and
                  Jianxin Wang},
  title        = {Parameterized Complexity of Multi-Winner Determination: {M}ore Effort Towards Fixed-Parameter Tractability},
  booktitle    = {Proceedings of the 17th International Conference on Autonomous Agents and MultiAgent Systems ({AAMAS} 2018)},
  pages        = {2142--2144},
  year         = {2018}
}

\end{document}